\documentclass[a4paper,english,12pt]{article}
\usepackage[T1]{fontenc}
\usepackage[utf8]{inputenc}
\usepackage[ngerman,english]{babel}
\usepackage{lmodern}
\usepackage{graphicx}
\usepackage{amssymb}
\usepackage{amsthm}
\usepackage{amsmath}
\usepackage{amsfonts}
\usepackage{tikz-cd}
\usepackage{dsfont}
\usepackage{flafter}
\usepackage{placeins}
\usepackage{multido}
\usepackage{mathrsfs}
\usepackage{etoolbox}
\usepackage{needspace}
\usepackage{microtype}
\usepackage{enumerate}
\usepackage{enumitem}
\usepackage{color}
\usepackage{textcomp}
\usepackage{tabularx}
\usepackage{authblk}
\usepackage{hyperref}

\newcommand{\pabl}[2]{\frac{\partial #1}{\partial #2}}
\newcommand{\id}{\mathds{1}}
\newcommand{\eps}{\varepsilon}
\newcommand{\rar}{\rightarrow}

\newcommand{\lrar}{\longrightarrow}
\newcommand{\uli}{\underline}

\newcommand{\equi}{ \quad \Leftrightarrow \quad }

\newcommand{\matfett}[1]{\text{\boldmath$#1$\unboldmath}}

\theoremstyle{plain}
\newtheorem{Theorem}{Theorem}[section]

\newtheorem{Lemma}[Theorem]{Lemma}

\theoremstyle{definition}
\newtheorem{Def}[Theorem]{Definition}

\theoremstyle{remark}

\AtBeginEnvironment{Definition}{\Needspace{5\baselineskip}}
\AtBeginEnvironment{Def}{\Needspace{5\baselineskip}}
\AtBeginEnvironment{Theorem}{\Needspace{5\baselineskip}}
\AtBeginEnvironment{Satz}{\Needspace{5\baselineskip}}
\AtBeginEnvironment{Lemma}{\Needspace{5\baselineskip}}
\AtEndEnvironment{Definition}{\Needspace{5\baselineskip}}
\AtEndEnvironment{Def}{\Needspace{5\baselineskip}}
\AtEndEnvironment{Theorem}{\Needspace{5\baselineskip}}
\AtEndEnvironment{Satz}{\Needspace{5\baselineskip}}
\AtEndEnvironment{Lemma}{\Needspace{5\baselineskip}}
\AtEndEnvironment{Note}{\Needspace{5\baselineskip}}

\numberwithin{equation}{section}

\begin{document}
\selectlanguage{english}

\title{Quantization of Perturbations in Inflation}
\author[1,2]{Benjamin Eltzner}
\affil[1]{Institut f\"ur Theoretische Physik, Universit\"at Leipzig,\authorcr
Postfach 100 920, D-04009 Leipzig, Germany}
\affil[2]{Max-Planck-Institut for Mathematics in the Sciences,\authorcr
Inselstra{\ss}e 22, D-04103 Leipzig, Germany}
\date{July 2013}

\maketitle

\begin{abstract}
	The derivation of the angular spectrum of temperature perturbations of the cosmic microwave background relies on the quantization of field and metric perturbations in the inflationary phase. The quantization procedure thus deserves a close examination. As the background spacetime on which the degrees of freedom that are quantized live is curved, the methods of quantum field theory in curved spacetimes are applicable. Furthermore the dynamic system that is quantized contains a constraint, which adds an interesting problem to the quantization of the system. This article investigates the treatment of the constraints and the Weyl quantization of the system. A preferred quantization procedure is identified, which renders the field and metric perturbations non-local.
\end{abstract}

\section{Introduction}

The theoretic framework of linear cosmological perturbation theory has been used to predict the angular spectrum of temperature perturbations of the cosmic microwave background with great success. As the derivation of the perturbation spectrum relies on the assumption of an inflationary phase in the early universe, the success of the predictions provides a cornerstone argument for the assumption of an inflationary phase in the early universe.

The derivation of the perturbation spectrum starts out with a generic scalar field setting for inflation where a semiclassical treatment is employed. The semiclassical Einstein equation, see \cite{Wal77} for an early discussion, provides an approach which is widely used in the context of quantum field theory on curved spacetime. In this context, matter is assumed to be modeled by quantum field theory, while the spacetime geometry is treated classically. The matter contribution is then represented by the expectation value of the energy-momentum operator. In the case at hand a different setting is chosen, where both the metric and the field are split up into a classical background part and a small quantum perturbation, as e.g. in \cite{BFM}, \cite{Dod} and \cite{Str}.

As a rationale for this approach to quantize geometric quantities one can point to the very high energy density that is expected to govern the inflationary phase suggesting the need to account for quantum effects of gravity. Concerning questions surrounding the choice of quantum state and its impact on the spectrum of perturbations we refer the reader to \cite{Hol}, \cite{MM}, \cite{BMR}, \cite{KKLSS}, \cite{Kun} and references therein. For the derivation of the angular temperature spectrum from the quantum perturbations see \cite{BFM}, \cite{Dod} and \cite{Str}, for a critical assessment see \cite{PSS}.

In the present article the quantization of the dynamic system of perturbations during inflation will be investigated. As preparation, section \ref{sec:basics} introduces the model that is investigated in the following. Section \ref{sec:symfor} is dedicated to the investigation of the dynamic system of perturbations in a phase space setting using standard tools for dynamic systems with constraints. The Dirac bracket governing the classical dynamics is identified and in section \ref{sec:dynvar} an argument is proposed, which enables a unique choice of dynamic variables. The preferred quantization procedure implies non-locality of the field and metric perturbations which provides an interesting connection to non-commutative inflation \cite{ABM}. Lastly, Weyl quantization of the dynamic system with constraints is discussed in section \ref{sec:weylq}.

\section{Linear Cosmological Perturbation Theory} \label{sec:basics}

Cosmological perturbation theory is a special case of perturbation theory in general relativity. The present article will only be concerned with linear perturbation theory, which means that only the tree level and the first order in the perturbations are considered and higher orders are neglected. For this to be a sensible approximation, it is important to make sure that solutions of the linearized equations of motion are linear approximations to exact solutions and that the impact of higher order perturbations is negligible. The first requirement has become known as linearization stability and certain necessary relations involving the second order perturbations were formulated \cite{Mon75}\cite{Mon76}. It has been shown that large classes of well-behaved solutions exist for Robertson-Walker background models \cite{DEa}. However, the topic has recently been revisited by Losic and Unruh \cite{LU08} and it was found that the impact of the stability conditions for background models close to de Sitter spacetime, which are investigated in the present article, are of the same order of magnitude as the linear perturbations. This fact makes an investigation in a more general context of nonlinear perturbation theory appear attractive. Allthough we do not treat higher orders of perturbation theory, we will point out which differences occur in that case.

The basic idea of linear perturbation theory is to split all quantities up into a background part and a perturbation. In the case of geometric quantities, most prominently the metric, this is nontrivial due to the role of diffeomorphism covariance in general relativity. It turns out that only part of this covariance survives in the setting of linear perturbation theory, which is interpreted as covariance under infinitesimal diffeomorphisms and labeled gauge freedom.

Using a $3+1$ decomposition of the background metric, as reviewed in \cite{Wal84}, it is possible to classify linear perturbations by their transformation behavior under symmetries of the spatial section. This leads to the definition of scalar, vector and tensor perturbations. In the present article we will only be concerned with scalar perturbations, thus vector and tensor perturbations will be dropped, after they have been identified.

The treatment in this section mostly follows \cite{Str} where no other sources are cited.

The background spacetime is assumed to be have a Friedmann-Robertson-Walker metric with flat spatial sections with line element
\begin{align*}
	ds^2 = a^2(\eta) \left(d\eta^2 - d\vec{x}^2\right) .
\end{align*}
Here, $\eta$ is the conformal time and the spacetime is obviously conformally flat. For some given quantity $X$ the derivative with respect to conformal time is denoted as $X'$. We also denote $\mathcal{H}=\tfrac{a'}{a}$. The time dependence of the scale factor $a$ depends crucially on the matter model coupled to the spacetime via the Einstein equations.

As preparation for the formulation of the considered model, a short introduction to linear perturbation theory is in order. The presentation follows the formulation used by Moncrief \cite{Mon75}\cite{Mon76} which slightly differs in style from what is usually found in the literature, e.g. \cite{BFM} and \cite{Str}. The key element for the description of linear perturbations and gauge freedom is the interpretation of perturbed quantities as one-parameter families of tensor fields.

\begin{Def} $ $
	\begin{itemize}
		\item Let $\mathcal{O} \subset \mathbb{R}$ an open neighborhood of $0$ and let $T : \mathcal{O} \rar \Gamma(T^r_s\mathcal{M})$ a smooth one-parameter family of smooth r-s-tensor fields. Then $T_0:=T(0)$ is called the background quantity and $\delta T := \tfrac{d}{d\lambda} T(\lambda)|_{\displaystyle {}_{\lambda=0}} \lambda$ is called the linear perturbation.
		\item Let $\Phi : \mathcal{O} \rar \textnormal{Diff} \mathcal{M}$ a smooth one-parameter family of diffeomorphisms generated by the vector field $X$, i.e. $\Phi(0) = \textnormal{id}_\mathcal{M}$, $X = \tfrac{d}{d\lambda} \Phi(\lambda)|_{\displaystyle {}_{\lambda=0}}$ and for the pull-back $(T^r_s\Phi)(0) = \textnormal{id}_{\Gamma(T^r_s\mathcal{M})}$. Then we define $\widetilde{\delta T} := \tfrac{d}{d\lambda} \big( (T^r_s\Phi)(\lambda) T(\lambda)\big)|_{\displaystyle {}_{\lambda=0}} \lambda$ as the transformed perturbation while the background quantity is unchanged, $\tilde{T}_0 := (T^r_s\Phi)(0) T_0 = T_0$.
	\end{itemize}
\end{Def}

The above definition has the benefit that it is mathematically concise and yields the transformation law of the perturbations under gauge transformations in a very transparent way,
\begin{align*}
	\widetilde{\delta T} = \delta T + \lambda \mathscr{L}_X T_0 =: \delta T + \mathscr{L}_{\xi} T_0 .
\end{align*}
In this equation $\xi = \lambda X$ is often interpreted as an ``infinitesimal vector field'' to justify the linearization. In any case the linearization is seen to be essentially a linearization of the Taylor expansion in $\lambda$ of one-parameter families of tensor fields. Using the Levi-Civita connection and local coordinates with respect to the background metric $g_0$ we get
\begin{align*}
	\widetilde{\delta T}^{\matfett{\rho}}{}_{\matfett{\sigma}} = \delta T^{\matfett{\rho}}{}_{\matfett{\sigma}} + \xi^\alpha \nabla_\alpha T_0{}^{\matfett{\rho}}{}_{\matfett{\sigma}} - \sum_{n=1}^r (\nabla_\alpha \xi^{\rho_n}) T ^{\,\hat{\!\matfett{\rho}}_n}{}_{\matfett{\sigma}} + \sum_{n=1}^s (\nabla_{\sigma_n} \xi^\alpha) T ^{\matfett{\rho}}{}_{\hat{\matfett{\sigma}}_n}
\end{align*}
where the multiindices $\hat{\matfett{\sigma}}_n$ and $\,\hat{\!\matfett{\rho}}_n$ are equal to the multiindices $\matfett{\sigma}$ and $\matfett{\rho}$ respectively, except that the index at the $n$-th position, i.e. $\sigma_n$ or $\rho_n$ respectively, is replaced by $\alpha$. Due to the metric compatibility of the Levi-Civita connection this yields an especially simple transformation law for the metric perturbation $\delta g$ given by
\begin{align*}
	\widetilde{\delta g}_{\mu\nu} = \delta g_{\mu\nu} + \nabla_{\mu} \xi_\nu + \nabla_{\nu} \xi_\mu .
\end{align*}

In the case of nonlinear perturbations it may be reasonable to also consider nonlinear gauge transformations, which makes the treatment of gauge transformations much more complex. However, an investigation of higher order perturbations in the context of the Dirac formalism should strive to treat gauge freedom in terms of first class constraints. Thus the difficulties of an explicit formulation of the gauge freedom would not play a role at this stage of the treatment.

It is customary to classify metric perturbations using a $3+1$ decomposition of the background spacetime into spatial section and time direction if the background metric is sufficiently symmetric, which is the case here. The perturbations are decomposed into scalar, vector and tensor part, dependent on their transformation behavior under the symmetry group of the spatial section.

One naively gets 4 scalar, 4 vector and 2 tensor perturbation degrees of freedom. However, the vector fields $\xi$ (generating gauge transformations) can be decomposed into a scalar part $\xi_s = \xi^0 \partial_0 + \xi^{|i} \partial_i$ and a vector part $\xi_v = x^i \partial_i$ with $x^i_{|i} = 0$, with two components each. This means that the gauge freedom reduces the degrees of freedom such that 2 gauge invariant degrees of freedom of every type remain. The different types of perturbations do not mix under gauge transformations and dynamical evolution, provided the spatial sections of the background spacetime satisfy certain curvature constraints \cite{Lo05}. In higher orders of perturbation theory the different types of perturbations are dynamically coupled, which reduces the usefulness of this decomposition in those cases \cite{Lo05}.

In the following the necessary formulae for a basic treatment of cosmological questions in the context of linear perturbation theory are derived. As explained above, the background metric used for cosmological perturbation theory is a Robertson-Walker metric with flat spatial section. We restrict now to scalar perturbations, so the perturbed metric can be expressed using conformal time as
\begin{align*}
	ds^2 = a^2(\eta) \Big( (1+2A) d\eta^2 + 2 B_{|i} d\eta dx^i - \big( (1+2D)\delta_{ij} + 2E_{|ij} \big) dx^i dx^j \Big)
\end{align*}
where $A$, $B$, $D$ and $E$ are functions of all coordinates. It should be noted that $A$, $B$, $D$ and $E$ are not exactly the same as the $s_n$ used above, but differ by a factor of $2 a^2(\eta)$. This yields slightly different behavior under gauge transformations. As we restrict to scalar perturbations we also restrict to scalar gauge vector fields $\uli{\xi} = \xi^0 \partial_0 + \xi^{|i} \partial_i$. This yields transformation laws
\begin{align*}
	\tilde{A} =& A + \mathcal{H} \xi^0 + (\xi^0)' & \tilde{B} =& B + \xi^0 - \xi'\\
	\tilde{D} =& D +\mathcal{H} \xi^0 & \tilde{E} =& E + \xi
\end{align*}
Obviously any two of the perturbations can be brought to vanish simultaneously except $A$ and $D$, as they are both only changed due to $\xi^0$. This means if $A$ and $D$ can be brought to vanish simultaneously in some gauge, this has a physical implication. An often used gauge is the conformal gauge, where $B_c=E_c=0$. Gauge invariant variables are frequently used to separate physical contents from gauge effects. However, it was shown by Unruh \cite{Un98} some time ago, that for each gauge fixing a set of gauge invariant variables satisfying the same dynamic equations exists, so some caution regarding physical interpretation of results is required. As we have four scalar perturbations and two gauge parameters we can derive two gauge invariant scalar perturbations. The ones usually defined in the literature are the Bardeen potentials
\begin{align*}
	\Psi =& A - \frac{1}{a} \big( (B+E') \big)' & \Phi =& D - \mathcal{H} (B+E')
\end{align*}

As matter model, we use a simple scalar field $\phi$ with a so far unspecified self-interaction potential $V$. The field satisfies the Klein-Gordon field equation
\begin{align*}
	\square \phi = - \pabl{V}{\phi}(\phi) = : - V_{,\phi}
\end{align*}
where $V(\phi)$ is some potential. The energy-momentum tensor takes the form
\begin{align*}
	T^\mu_\nu = \partial^\mu \phi \partial_\nu \phi -\delta^\mu_\nu \Big( \frac{1}{2} \partial^\lambda \phi \partial_\lambda \phi - V(\phi) \Big)
\end{align*}

In keeping with the idea of linear perturbation theory also the scalar field is decomposed into a background part and a perturbation $\phi = \phi_0 + \delta \phi$. The gauge transformation behavior of the field perturbation is
\begin{align*}
	\widetilde{\delta\phi} = \delta\phi + \xi^0 \phi'_0
\end{align*}
which implies the following gauge invariant perturbation
\begin{align*}
	\chi = \delta \phi - \phi'_0 (B+E')
\end{align*}
Again in conformal gauge $\delta \phi_c = \chi$; in the following the subscript indicating conformal gauge will be suppressed and whenever expressions contain $\delta \phi$ they are to be understood in conformal gauge.

To formulate the Einstein equation it is necessary to calculate the Einstein tensor, which is not gauge invariant. However $G^\mu_\nu -8 \pi G T^\mu_\nu$ can be cast into a gauge invariant form, as it vanishes in any gauge. This implies that the Einstein equations can be formulated using only gauge invariant variables. This means we can calculate the Einstein tensor and the energy momentum tensor in some arbitrary gauge, and the resulting Einstein equations will be true for any gauge. Hence, we can replace the perturbation variables in the Einstein equations by gauge invariant variables that are equal to them in the chosen gauge. The resulting equation, containing only gauge invariant variables, still holds in any gauge. A calculation using the above strategy is conducted in \cite{Str} and we will only give the results here.

The idea of perturbation theory is to assume the equations to hold order by order. This means we will have a background Einstein and Klein-Gordon equations and perturbations to these equations which will be assumed to hold separately. The background Einstein equations are
\begin{align*}
	- 3 \mathcal{H}' =& 4\pi G (\phi'_0)^2 + 8 \pi G V(\phi_0) a^2\\
	- \mathcal{H}' - 2 \mathcal{H}^2 =& - 4 \pi G (\phi'_0)^2 + 8 \pi G V(\phi_0) a^2
\end{align*}
and the background Klein-Gordon equation reduces to
\begin{align}
	\phi''_0 + 2 \mathcal{H} \phi'_0 + V_{,\phi}(\phi_0)a^2 =0 \label{bgKGG}
\end{align}

The perturbation Klein-Gordon and Einstein equations will lead to a set of four independent equations. From the Einstein equations we get
\begin{align}
	\Phi &= - \Psi\label{lEin0}\\
	\Psi'' &- \triangle \Psi + 2 \left( \mathcal{H} - \frac{\phi''_0}{\phi_0'} \right) \Psi' + 2 \left( \mathcal{H}' - \frac{\phi''_0}{\phi_0'} \mathcal{H} \right) \Psi = 0 \label{lEin1}\\
	\Psi' + \mathcal{H} \Psi &= 4 \pi G \phi'_0 \chi \label{lEin2}
\end{align}
and the first order of the Klein-Gordon equations yields
\begin{align}
	\chi'' - \triangle \chi + 2 \mathcal{H} \chi' + V,_{\phi\phi}(\phi_0) a^2 \chi = 4 \phi_0' \Psi' - 2 V,_{\phi} (\phi_0) a^2 \Psi \label{lKGG}
\end{align}

The coupled differential equations \eqref{lEin1}, \eqref{lKGG} and \eqref{lEin2} govern the dynamics of the gauge invariant perturbations $\Psi$ and $\chi$. The constraint equation \eqref{lEin2} only holds up to a homogeneous term, which plays an important role in the following. It should be noted, that the dynamic system is significantly more complex in the case that higher order perturbations are considered. The decomposition into scalar, vector and tensor perturbations and the usage of gauge invariant variables cannot be straightforwardly generalized and nonlinearities in the linear perturbations occurring in the dynamic equations increase complexity.

The present article will investigate gauge invariant linear perturbations using these equations and $\Phi$ will not be treated as an independent variable due to equation \eqref{lEin0}. An alternative derivation of the dynamic system is possible by means of a Lagrangian approach reviewed in \cite{BFM}. The Lagrangian approach will be investigated in section \ref{sec:symfor}.

\section{Canonical Symplectic Form} \label{sec:symfor}

It is clear from the dynamic system  consisting of \eqref{lEin1}, \eqref{lKGG} and \eqref{lEin2} that only one independent scalar degree of freedom exists. From the constraint equation \eqref{lEin2} one can already anticipate that $\Psi$ and $\chi$ cannot both satisfy canonical commutation relations. As we are not aware that this has been done previously, we will make this point explicit here.

For simplified and homogeneous display of the equations, the following calculations use the parameter $z=\tfrac{a\phi'_0}{\mathcal{H}}$ and the slow-roll parameters
\begin{align}
	\eps =&\; 1 - \frac{\mathcal{H}'}{\mathcal{H}^2} = 4 \pi G \frac{z^2}{a^2} \label{defeps}\\
	\delta =&\; 1 + \eps - \frac{z'}{z\mathcal{H}} \label{defdelta}
\end{align}

These relations can be used to unify the notation for the dynamic system such that only $a$, $\mathcal{H}$, $\delta$, $\eps$ and $z$ appear in the coefficients. Recasting all terms using $\eps$ and $\delta$ sometimes enlarges terms but is nevertheless performed in order to stick to a unified formulation.
\begin{alignat}{4}
	\Psi'' - && \triangle \Psi + && 2 \delta \mathcal{H} \Psi' + && 2 ( \delta - \eps ) \mathcal{H}^2 \Psi & = 0 \label{E01}\\
	&& && \Psi' + && \mathcal{H} \Psi & = \frac{\eps \mathcal{H}a}{z} \chi \label{E02}
\end{alignat}

Using the relations
\begin{align}
	\chi &= \frac{z}{\eps \mathcal{H} a} \Psi' + \frac{z}{\eps a} \Psi \label{constr1}\\
	\chi' &= - \frac{\delta z}{\eps a} \Psi' + \frac{z}{\eps \mathcal{H} a} \big( \triangle + ( \eps - \delta ) \mathcal{H}^2 \big) \Psi \label{constr2}
\end{align}
the following result can be calculated

\begin{Lemma} \label{Lemma:CCRNoGo} $ $ \\
	If there is a homogeneous bidistribution with integral kernel $f(\vec{x}-\vec{y})$ and the following relations hold
	\begin{align}
		\left[\Psi(x), \Psi(y) \right]\big|_{x^0=y^0} &= \left[\Psi'(x), \Psi'(y) \right]\big|_{x^0=y^0} = 0\\
		\left[\Psi(x), \Psi'(y) \right]\big|_{x^0=y^0} &= \left[\Psi(y), \Psi'(x) \right]\big|_{x^0=y^0} = f(\vec{x}-\vec{y})
	\end{align}
	then
	\begin{align}
		\left[\chi(x), \chi(y) \right]\big|_{x^0=y^0} &= \left[\chi'(x), \chi'(y) \right]\big|_{x^0=y^0} = 0\\
		\left[\chi(x), \chi'(y) \right]\big|_{x^0=y^0} &= - \frac{z^2 \big( \triangle^x + \eps \mathcal{H}^2 \big) }{(\eps \mathcal{H} a)^2} \left[\Psi(x), \Psi'(y) \right]\big|_{x^0=y^0}
	\end{align}
	Note that $a$, $\eps$, $\mathcal{H}$ and $z$ only depend on time and the expressions investigated here are evaluated at equal times. Therefore their time dependence is suppressed.
\end{Lemma}

It is interesting to derive a dynamic equation for $\chi$ alone from \eqref{lKGG} using the constraints. The resultant equation
\begin{align}
 &(\triangle + \eps\mathcal{H}^2) \chi'' + 2 \mathcal{H}(\triangle + \eps\delta\mathcal{H}^2) \chi' + \left( - \triangle^2 + \left(- 2 \eps + 3 \delta - \eps \delta - \delta^2 + \frac{\delta'}{\mathcal{H}} \right) \mathcal{H}^2 \triangle  \right. \nonumber\\
 &\qquad \left. + \left(- \eps + \delta - \eps \delta + \delta^2 + \frac{\delta'}{\mathcal{H}} \right) \eps \mathcal{H}^4 \right) \chi = 0 \label{E07}
\end{align}
is not a normal Klein-Gordon equation but a dynamic equation of fourth order. Thus it is not to be expected that the CCR are of any significance for $\chi$. The field $\Psi$, however, fulfills a Klein-Gordon type equation, so at first glance it seems reasonable to assume standard CCR for this field as a means of quantization and treating $\chi$ as a derived field whose commutation relations as derived from the CCR for $\Psi$ bear no significance. However, such a naive approach will be proven misguided by the following investigation.

The usual choice of the field to quantize via CCR is $u = z \Psi + a \chi$ which is justified by the fact that the quantity $\mathcal{R} = - \tfrac{u}{z}$ can be interpreted as gauge invariant curvature perturbation, which is the quantity that is used to derive the temperature perturbation. In the present treatment we take the point of view that the physical significance of $u$ does not by itself justify choosing it as canonical variable for quantization. Therefore we are looking for a structural argument to identify a preferred quantization procedure.

In this section the dynamic system of the spatial Fourier modes is reviewed from the angle of symplectic geometry. To this avail, the dynamical system has to be cast in Hamiltonian form, which implies the identification of the Darboux coordinates of the symplectic form. To achieve this, it is necessary to use the action presented e.g. in \cite{BFM} and to derive the corresponding Hamiltonian system. Then the impact of the constraint on the symplectic structure is discussed. The following discussion will employ Hamiltonian and Lagrangian densities, where the term ``density'' will always be suppressed.

For the treatment of a dynamical system with second class constraints several methods exist. The dynamical variables, including the Lagrange multipliers, may be combined into a reduced set of dynamical variables in such a way that the resulting Lagrangian contains no constraints and only depends on the new variables. This is the strategy which is pursued in the standard literature on cosmological perturbations during inflation, especially \cite{BFM}. However, finding an appropriate reduced set of dynamical variables is not a straightforward task and is usually done by simply choosing the appropriate ansatz to begin with. As the resultant dynamical system is a generic Hamiltonian system, quantization is straightforward.

A second method consists in the straightforward construction of a Hamiltonian containing the Lagrange multipliers and the subsequent application of a standard procedure to eliminate the Lagrangian multipliers. This procedure does not lead to an obvious reduction of phase space, however it can be seen that the phase space is not symplectic and the complement of the maximally symplectic subspace has a dimension which is identical to the number of constraints. This procedure has the disadvantage that the geometry of the phase space is unclear and therefore quantization of the dynamical system is not straightforward.

Yet another method is the Dirac method, which will be employed in this work. In this context the constraints are ignored at first, to produce an ``unconstrained Hamiltonian'' with a corresponding canonical symplectic form. Then the constraints are used to derive the so-called Dirac bracket from the Poisson bracket of the unconstrained system. The tensor corresponding to the Dirac bracket will be called the ``Dirac form'' in the present treatment. The Dirac form is degenerate in a number of dimensions corresponding to the number of constraints and symplectic in the remaining dimensions. The present treatment will pursue the derivation of the Dirac bracket to shed some light on the quantization of the system. The classic reference for this method is \cite{Dir} but it is treated to varying degree of detail in several textbooks, e.g. \cite{Der}, \cite{HT} and \cite{Wei}.

We would like to point to \cite{DLP}, where a similar scenario is investigated but vector and tensor perturbations are also taken into account. However, that work applies a gauge fixing procedure, which is incompatible to the gauge invariant variables considered here, to achieve a reduced phase space, which leads to a different treatment than pursued in the present work.

As a first step it is convenient to rescale the fields $\Psi$ and $\chi$ so they have the same dimension. Therefore we define the rescaled fields
\begin{align*}
	\matfett{\psi} := z \Psi \qquad \lambda := a \chi
\end{align*}
The Lagrangian given by equation (10.68) in \cite{BFM} is given here in terms of the rescaled variables, suppressing the total divergences and simplifying by setting $\Phi = - \Psi$ and $L:= 4\triangle (B-E')$ (note the different notation). Note that the zero mode of $L$ vanishes, which implies that the constraint does not hold in the zero mode. The fact, that the constraint does not apply to the zero mode has first been pointed out by Unruh \cite{Un98} who also studied the dynamics of the zero mode in detail. A separate investigation of the zero mode will thus be performed in the following. From the treatment in \cite{BFM} it is clear that the terms introduced by switching to gauge invariant variables cancel out.
\begin{align*}
	\mathscr{L} =&\; \frac{1}{2\eps} \left[ - 3 (\matfett{\psi}')^2 + 6 (\eps-\delta)\mathcal{H} \matfett{\psi}' \matfett{\psi} + (\eps - 3 \eps^2 + 6 \eps \delta - 3 \delta^2) \mathcal{H}^2 \matfett{\psi}^2 + \matfett{\psi} \triangle \matfett{\psi} \right]\\
	&\; + \frac{1}{2} \left[ (\lambda')^2 - 2 \mathcal{H} \lambda' \lambda + \left(1- 3 \eps - 3 \delta + \eps \delta + \delta^2 - \frac{\delta'}{\mathcal{H}} \right) \mathcal{H}^2 \lambda^2 + \lambda \triangle \lambda \right]\\
	&\; + 4 \mathcal{H} \matfett{\psi}' \lambda + 2(1- 2 \eps + \delta) \mathcal{H}^2 \matfett{\psi} \lambda + z L \left[ - \frac{1}{\eps} \matfett{\psi}' + \frac{\eps -\delta}{\eps} \mathcal{H} \matfett{\psi} + \mathcal{H} \lambda \right]
\end{align*}
The fields are treated as real fields here, but a generalization to complex fields is straightforward.

The next step in the treatment is to consider the ``unconstrained Lagrangian'', thus setting $L=0$. We will additionally switch to a ``mode Lagrangian'', using modes $\matfett{\psi}_k : = \int e^{i \vec{k}\vec{x}} \matfett{\psi}(\vec{x})$. The modes only depend on $k = |\vec{k}|$ due to the fact that the dynamics are spatially homogeneous and isotropic.
\begin{align*}
	\mathscr{L}_{0,k} =&\; \frac{1}{2\eps} \left[ - 3 (\matfett{\psi}_k')^2 + 6 (\eps-\delta)\mathcal{H} \matfett{\psi}_k' \matfett{\psi}_k + [(\eps - 3 \eps^2 + 6 \eps \delta - 3 \delta^2) \mathcal{H}^2 - k^2] \matfett{\psi}_k^2 \right]\\
	&\; + \frac{1}{2} \left[ (\lambda_k')^2 - 2 \mathcal{H} \lambda_k' \lambda_k + \left[\left(1- 3 \eps - 3 \delta + \eps \delta + \delta^2 - \frac{\delta'}{\mathcal{H}} \right) \mathcal{H}^2 - k^2 \right] \lambda_k^2 \right]\\
	&\; + 4 \mathcal{H} \matfett{\psi}_k' \lambda_k + 2(1- 2 \eps + \delta) \mathcal{H}^2 \matfett{\psi}_k \lambda_k
\end{align*}
Note that the mode Lagrangian is not the Fourier transform of the Lagrangian, as a Fourier transformation would not map products of functions to products of their Fourier transforms but to convolutions. The spatial integral of the Lagrangian coincides with the mode integral of the mode Lagrangian, which means they lead to the same Lagrangian function and thus to the same dynamics.\footnote{This is essentially an application of the Fourier-Plancherel theorem.} Due to the fact that the dynamic equations do not mix modes, no mode-mixing occurs in the Lagrangian.

From the mode Lagrangian one calculates the canonical momentum modes
\begin{align*}
	\Pi_k &= \pabl{\mathscr{L}_{0,k}}{\matfett{\psi}_k'} = -\frac{3}{\eps} \matfett{\psi}_k' + 3\frac{\eps -\delta}{\eps} \mathcal{H} \matfett{\psi}_k + 4 \mathcal{H} \lambda_k\\
	\Leftrightarrow \quad \matfett{\psi}_k' &= -\frac{\eps}{3} \Pi_k + (\eps -\delta) \mathcal{H} \matfett{\psi}_k + \frac{4}{3} \eps \mathcal{H} \lambda_k \displaybreak[0]\\
	\kappa_k &= \pabl{\mathscr{L}_{0,k}}{\lambda_k'} = \lambda_k' - \mathcal{H} \lambda_k\\
	\Leftrightarrow \lambda_k' &= \kappa_k + \mathcal{H} \lambda_k
\end{align*}
which allows the derivation of the unconstrained Hamiltonian for the modes
\begin{align*}
	\mathscr{H}_{0,k} =&\; -\frac{\eps}{6} \Pi_k^2 + (\eps-\delta) \mathcal{H} \Pi_k \matfett{\psi}_k + \left( - \frac{1}{2} \mathcal{H}^2 + \frac{1}{2 \eps} k^2 \right) \matfett{\psi}_k^2 \\
	&\; + \frac{1}{2} \kappa_k^2 + \mathcal{H} \kappa_k \lambda_k + \left[\frac{1}{6} \left(- 7 \eps + 9 \delta - 3 \eps \delta - 3 \delta^2 + 3 \frac{\delta'}{\mathcal{H}} \right) \mathcal{H}^2 + k^2 \right] \lambda_k^2\\
	&\; + \frac{4}{3} \eps \mathcal{H} \Pi_k \lambda_k - 2(1 - \delta) \mathcal{H}^2 \matfett{\psi}_k \lambda_k
\end{align*}

The unconstrained phase space of Fourier transformed initial values is assumed to be $\big(\mathscr{S}(\mathbb{R}^+) \big)^{\times 4}$ in the following, which means that initial values are Schwartz functions in the spatial coordinates. The Hamiltonian density at fixed time is then a smooth functional $\mathscr{H} : \big(\mathscr{S}(\mathbb{R}^+) \big)^{\times 4} \lrar \mathscr{S}(\mathbb{R}^+)$. The initial value phase space for a single mode is $\mathbb{R}^4$ because the modes are simple functions of time. The mode Hamiltonian is thus a smooth function $\mathscr{H}_{0,k} : \mathbb{R}^4 \lrar \mathbb{R}$.

Any smooth phase space ``density'' $\mathscr{F} : \big(\mathscr{S}(\mathbb{R}^+) \big)^{\times 4} \lrar \mathscr{S}(\mathbb{R}^+)$ can be integrated to yield a smooth phase space ``function'' $F : \big(\mathscr{S}(\mathbb{R}^+) \big)^{\times 4} \lrar \mathbb{R}$. Furthermore, for a polynomial $f: \mathbb{R}^+ \times \mathbb{R}^4 \rar \mathbb{R}$ with $\forall k \in \mathbb{R}^+ : \, f(k,\vec{0}) = 0$, a smooth density $\mathscr{F}$ can be constructed as
\begin{align}
	\mathscr{F} \big( g_1, g_2, g_3, g_4 \big) (k) := f_k \big( g_1(k), g_2(k), g_3(k), g_4(k) \big)
\end{align}
where $\mathscr{F} \big( g_1, g_2, g_3, g_4 \big) \in \mathscr{S}(\mathbb{R}^+)$ if $g_1, g_2, g_3, g_4 \in \mathscr{S}(\mathbb{R}^+)$. Let us call the space of such polynomial densities $ \mathcal{D}_{\text{Pol},4} \big(\mathscr{S}(\mathbb{R}^+) \big) $ and the functions derived from such densities by integration as $ \mathcal{F}_{\text{Pol},4} \big(\mathscr{S}(\mathbb{R}^+) \big) $. As such polynomials encompass the mode Hamiltonian and the field modes, this suffices for our purpose. It should be noted that for a general smooth function $f: \mathbb{R}^+ \times \mathbb{R}^4 \rar \mathbb{R}$ with $\forall k \in \mathbb{R}^+ : \, f(k,\vec{0}) = 0$ a smooth function $F : \big(\mathscr{S}(\mathbb{R}^+) \big)^{\times 4} \lrar \mathbb{R}$ does not necessarily exist.

Now consider the Poisson bracket on a single mode
\begin{align}
	\{,\}_k : C^\infty \big( \mathbb{R}^4 \rar \mathbb{R} \big) \times C^\infty \big( \mathbb{R}^4 \rar \mathbb{R} \big) \lrar \mathbb{R}
\end{align}
The Poisson bracket for a phase space of several modes is always diagonal in the modes due to the symmetry of the dynamic system. It is obvious that this bracket can be lifted to densities created from polynomials as
\begin{align}
	\{,\}_{\mathcal{D}} &: \mathcal{D}_{\text{Pol},4} \big(\mathscr{S}(\mathbb{R}^+) \big) \times \mathcal{D}_{\text{Pol},4} \big(\mathscr{S}(\mathbb{R}^+) \big) \lrar \mathscr{S}(\mathbb{R}^+) \nonumber\\
	\{\mathscr{F},\mathscr{G}\}_{\mathcal{D}}(k) &:= \{\mathscr{F}(k),\mathscr{G}(k)\}_k
\end{align}
and thus a Poisson bracket on $ \mathcal{F}_{\text{Pol},4} \big(\mathscr{S}(\mathbb{R}^+) \big) $ is defined as
\begin{align}
	\{ F, G\} = \int_0^\infty \limits \{\mathscr{F},\mathscr{G} \}_{\mathcal{D}}(k) 4 \pi k^2 dk = \int_0^\infty \limits \{\mathscr{F}_k,\mathscr{G}_k \}_k 4 \pi k^2 dk
\end{align}
Note that we use the three dimensional volume element $4 \pi k^2 dk$ to account for the fact that $\vec{k}$ is actually a three dimensional vector. Only one integral is performed because the mode Poisson bracket is diagonal in the modes.

As the canonical variables of the unconstrained Hamiltonian are the Darboux coordinates of the canonical symplectic form on the unconstrained phase space, the symplectic form corresponding to the mode Poisson bracket takes the standard form
\begin{align*}
	\forall \eta : \; (\omega_k^{\alpha\beta})_{\alpha,\beta}(\eta) = \left( \begin{array}{cc} 0_2 & \id_2 \\ - \id_2 & 0_2 \end{array} \right)
\end{align*}
on the linear space of first derivatives of mode functions $\mathscr{F}_k : \mathbb{R}^4 \lrar \mathbb{R}$.

For the $k \neq 0$ modes the constraint equation implies a projection to a subspace of the phase space. Contrary to the projection to an energy hypersurface this projection constrains the accessible phase space for all solutions of the dynamic system so it can be interpreted as a genuine reduction of phase space. In canonical variables the primary constraint derived from equation \eqref{constr1}  reads
\begin{align}
	\mathcal{C}_1 = \Pi_k - \mathcal{H} \lambda_k = 0 \label{CM1}
\end{align}
and a secondary constraint can be derived as
\begin{align}
	\mathcal{C}_2 := - \frac{1}{\mathcal{H}} \{ \mathcal{C}_1, \mathscr{H}_{0,k} \} - \frac{1}{\mathcal{H}} \pabl{\mathcal{C}_1}{\eta} = \kappa_k - \frac{\eps\mathcal{H}^2 - k^2}{\eps \mathcal{H}} \matfett{\psi}_k + \delta \mathcal{H} \lambda_k = 0 \label{CM2}
\end{align}
using the first constraint. The tertiary constraint vanishes
\begin{align*}
	\mathcal{C}_3 = \{ \mathcal{C}_2, \mathscr{H}_{0,k} \} + \pabl{\mathcal{C}_2}{\eta} = 0
\end{align*}
applying the first two constraints, so there are exactly two geometrically independent constraints, which means that the phase space dimension is reduced by two. It is worth emphasizing that these constraints are satisfied by all valid solutions of the dynamic system, so all solutions of the mode dynamic system can be described as living on a two dimensional reduced phase space.

For the treatment of constraints it is of great importance whether they are first or second class constraints. In the case discussed here, only second class constraints occur, due to the following lemma.
\begin{Lemma} $ $ \\
	The constraints \eqref{CM1} and \eqref{CM2} are second class constraints for $k \neq 0$. For $k \rar 0$ the Dirac form diverges no stronger than of order $k^{-2}$.
\end{Lemma}
\begin{proof} $ $ \\
	To fix notation let $(M_{ab})_{a,b}$ denote the matrix with components $M_{ab}$. If the constraint matrix $(C_{ab})_{a,b}:=\big(\{\mathcal{C}_a,\mathcal{C}_b\}\big)_{a,b}$ is invertible, the constraints are second class constraints. Using the definition of the Poisson bracket
	\begin{align*}
		\{f,g\} &:= \pabl{f}{\matfett{\psi}_k} \pabl{g}{\Pi_k} + \pabl{f}{\lambda_k} \pabl{g}{\kappa_k} - \pabl{f}{\Pi_k} \pabl{g}{\matfett{\psi}_k} - \pabl{f}{\kappa_k} \pabl{g}{\lambda_k}
	\end{align*}
	one gets
	\begin{align*}
		C_{12} = \{ \mathcal{C}_1, \mathcal{C}_2 \} = \frac{\eps\mathcal{H}^2 - k^2}{\eps \mathcal{H}} - \mathcal{H} = - \frac{k^2}{\eps \mathcal{H}}
	\end{align*}
	which yields, due to antisymmetry
	\begin{align*}
		(C_{ab})_{a,b} = \frac{k^2}{\eps \mathcal{H}} \left( \begin{array}{cc} 0 & -1 \\ 1 & 0 \end{array} \right)
	\end{align*}
	which is invertible as long as $k \neq 0$. As the inverse diverges as $k^{-2}$, also the second part of the claim is true.
\end{proof}

First class constraints usually occur in the context of gauge theories so one might expect that a full treatment of scalar perturbation theory, not using gauge invariant variables, might be significantly more complicated, especially when it comes to quantization. Such a more extensive investigation is not attempted in the present work.

In the present case, the constraints would only be first class constraints for the zero mode $k=0$, where they do not exist as pointed out above. Thus a complete formulation of the Dirac bracket on the full phase space of all modes is possible in a generic way.

The Dirac bracket is derived from the Poisson bracket and the constraint matrix as
\begin{align*}
	\{f,g\}_{\text{D}} := \{f,g\} - \{f,\mathcal{C}_i\} (C^{-1})_{ij} \{\mathcal{C}_j,g\}
\end{align*}
and a straightforward calculation yields the mode Dirac form
\begin{align}
	(\omega_{\text{D},k}^{\alpha\beta})_{\alpha,\beta} = \left( \begin{array}{cccc} 0 & \frac{\eps\mathcal{H}}{k^2} & \frac{\eps\mathcal{H}^2}{k^2} & - \frac{\eps\delta\mathcal{H}^2}{k^2} \\ -\frac{\eps\mathcal{H}}{k^2} & 0 & 0 & 1 - \frac{\eps\mathcal{H}^2}{k^2} \\ -\frac{\eps\mathcal{H}^2}{k^2} & 0 & 0 & \mathcal{H} - \frac{\eps\mathcal{H}^3}{k^2} \\ \frac{\eps\delta\mathcal{H}^2}{k^2} & - 1 +  \frac{\eps\mathcal{H}^2}{k^2} & - \mathcal{H} + \frac{\eps\mathcal{H}^3}{k^2} & 0 \end{array} \right) \label{DiracForm}
\end{align}
which obviously diverges as $k^{-2}$ for $k\rar 0$. The Dirac bracket on all modes is defined analogously to the Poisson bracket on all modes. On the linear space of first derivatives, we cast it in the symbolic form
\begin{align}
	\left\{ \left( \begin{array}{c} f_1 \\ f_2 \\ f_3 \\ f_4 \end{array} \right) , \left( \begin{array}{c} g_1 \\ g_2 \\ g_3 \\ g_4 \end{array} \right) \right\}_{\text{D}} := \int \left\{ \left( \begin{array}{c} f_1(\vec{k}) \\ f_2(\vec{k}) \\ f_3(\vec{k}) \\ f_4(\vec{k}) \end{array} \right) , \left( \begin{array}{c} g_1(\vec{k}) \\ g_2(\vec{k}) \\ g_3(\vec{k}) \\ g_4(\vec{k}) \end{array} \right) \right\}_{\text{D},k} d^3 k \label{IntDiracDef}
\end{align}
where all the functions $f_i$ and $g_i$, symbolizing the values of first derivatives of smooth densities $\mathscr{F} : \big(\mathscr{S}(\mathbb{R}^+) \big)^{\times 4} \lrar \mathscr{S}(\mathbb{R}^+)$, are Schwartz functions in $k$. As the mode Dirac form given in equation \eqref{DiracForm} behaves like $k^{-2}$ for $k\rar 0$ and the volume element of the integral in equation \eqref{IntDiracDef} gives a factor of $k^2$ the limit of the integrand for $k\rar 0$ is finite. As the constraint does not exist in the zero mode, the integrand is finite at $k=0$ and does not contribute to the integral. Thus the integrand is a Schwartz function and the integral is well defined.

The remarkable fact that the $k=0$ mode does not contribute to the structure of the dynamics on the phase space can be heuristically traced back to the fact, that only Schwartz initial values are considered in the present work. This prevents singularities at $k=0$, which would boost the importance of the zero mode. The restriction to Schwartz initial values does not pose technical obstacles to a meaningful quantum field description, therefore this restriction is considered acceptable here.

As a last step, the constraints can be factored out of the phase space, such that the resultant factor space can serve as physical phase space. This procedure will be investigated to some more depth in the context of quantization.

\section{Preferred Dynamic Variables} \label{sec:dynvar}

We now turn to investigating which canonical variables are suited for quantization of the dynamic system with constraint. The canonical variables can be chosen freely and only need to satisfy the requirement $\{X,P\}_{\text{D}} =1 $. This freedom implies the application of the constraints as well as symplectomorphisms. To make a connection to section \ref{sec:symfor}, the possible canonically conjugate variables $P$ for $X = u_k = \matfett{\psi}_k + \lambda_k$ and for $X = \tfrac{1}{z} \matfett{\psi}_k$ will be identified here.

Considering $X = u_k$ the corresponding conjugate variable has to satisfy
\begin{align*}
	-\frac{\eps \mathcal{H}}{k^2} \pabl{P}{\matfett{\psi}_k} + \frac{\eps \mathcal{H}}{k^2} \pabl{P}{\lambda_k} + \frac{\eps \mathcal{H}^2}{k^2} \pabl{P}{\Pi_k} + \frac{k^2 - \eps \mathcal{H}^2 - \eps \delta \mathcal{H}^2}{k^2} \pabl{P}{\kappa_k} = 1
\end{align*}
which is indeed fulfilled by
\begin{align*}
	P = u_k' = (\eps - \delta) \mathcal{H} \matfett{\psi}_k + \left( 1 + \frac{4}{3} \eps \right) \lambda_k - \frac{\eps}{3} \Pi_k + \kappa_k
\end{align*}
which implies that the standard quantization procedure is indeed justified.

For $X = \tfrac{1}{z} \matfett{\psi}_k$ the corresponding conjugate variable has to satisfy
\begin{align*}
	\frac{\eps \mathcal{H}}{z k^2} \pabl{P}{\lambda_k} + \frac{\eps \mathcal{H}^2}{z k^2} \pabl{P}{\Pi_k} - \frac{- \eps \delta \mathcal{H}^2}{z k^2} \pabl{P}{\kappa_k} = 1
\end{align*}
which can be fulfilled by
\begin{align*}
	P = \alpha X' = - \frac{\alpha \mathcal{H}}{z} \matfett{\psi}_k + \frac{4}{3} \frac{\alpha \eps \mathcal{H}}{z} \lambda_k - \frac{\alpha \eps}{3 z} \Pi_k
\end{align*}
namely if $\alpha = \tfrac{z^2 k^2}{\eps^2 \mathcal{H}^2}$. This result is peculiar in that it includes a prefactor of $k^2$. A comparison with equation \eqref{CCRrel} shows that the treatment of this subsection reproduces the result of the previous subsection with respect to the relation of the commutators (or in the classical setting Dirac brackets). However, the treatment in the present subsection allows to take the point of view that the standard treatment of quantizing $u$ is canonical in the sense of Dirac quantization, while the straightforward CCR quantization of $\Psi$ is not.

As the fields $\matfett{\psi}$ and $\lambda$ are scalar, it seems a generic assumption, that a reasonable canonical variable should have a scalar field interpretation in the sense that it satisfies a Klein-Gordon-type equation. Two simple technical requirements which imply a scalar field interpretation justify the choice of the field $u$ as canonical variable as can be seen in the following theorem. 

\begin{Theorem} \label{theorem:canvar} $ $ \\
	Let $A$ and $B$ be smooth real-valued functions of time but independent of $k$ and let $B$ be strictly positive. Assume $X_k = A \matfett{\psi}_k + B \lambda_k$ then $B=1$ and $A=1$ or $A=1 + \frac{(\eps - \delta)(1 + \eps - \delta)}{\eps (1+ \eps)}$ are the only solutions compatible the technical requirements
	\begin{enumerate}
		\item $P_k =X'_k$ with $\{X_k,P_k\}_{\text{D}} =1$
		\item the dynamic equation for $X$ is a differential equation of at most second order.
	\end{enumerate}
\end{Theorem}
\begin{proof} $ $\\
	Requirement (a) leads to the equation
	\begin{align}
		\frac{\eps \mathcal{H}}{k^2}(AB'-A'B) + \frac{k^2 -\eps \mathcal{H}^2}{k^2} B^2 + (1 -\eps) \frac{\eps \mathcal{H}^2}{k^2} AB + \frac{\eps^2 \mathcal{H}^2}{k^2} A^2 &= 1\\
		 \equi \eps \mathcal{H}(AB'-A'B) + (k^2 -\eps \mathcal{H}^2) B^2 + (1 -\eps) \eps \mathcal{H}^2 AB + \eps^2 \mathcal{H}^2 A^2 &= k^2 .
	\end{align}
	Because $B$ does not depend on $k$, $B^2 =1$ and thus $B =1$. This yields
	\begin{align}
		A' - \eps \mathcal{H} A^2 - (1 -\eps) \mathcal{H} A + \mathcal{H} = A' - (A - 1) (A \eps + 1) \mathcal{H} = 0
	\end{align}
	which can be cast in a simpler form using $D := A - 1$
	\begin{align}
		D' - D (D \eps + (1+\eps) ) \mathcal{H} = 0
	\end{align}
	such that one obtains a Bernoulli differential equation. This type of equation is solvable with trivial solution $D=0$ and a nontrivial solution which can be found by using $E = \tfrac{1}{D}$ leading to
	\begin{align}
		E' + (1+\eps) \mathcal{H} E = - \eps \mathcal{H}
	\end{align}
	a readily solvable linear equation. The solution for $A$ is
	\begin{align}
		A= \left( 1 - \frac{a^2}{\mathcal{H}} \left( \int_{\eta_0}^\eta \limits \eps (\eta') a^2(\eta') \, d \eta' + K \right)^{-1} \right)
	\end{align}
	where $K$ and $\eta_0$ account for one free constant.
	
	Now abbreviating $\rho := A-1$ one readily sees
	\begin{align}
		\rho' &= \eps \mathcal{H} \rho^2 + (1 + \eps) \mathcal{H} \rho \label{drho}\\
		\rho'' &= 2 \eps^2 \mathcal{H}^2 \rho^3 + 2 \eps (1 + \eps - \delta) \mathcal{H}^2 \rho^2 + 2 (1 + \eps + \eps^2 - \eps \delta) \mathcal{H}^2 \rho
	\end{align}
	
	Defining $u_{\rho,k} = u_{k,0} + \rho \matfett{\psi}_k$ the dynamic equation can be derived from
	\begin{align}
		u''_\rho = u''_0 + \rho'' \matfett{\psi} + 2 \rho' \matfett{\psi}' + \rho \matfett{\psi}''
	\end{align}
	using the dynamic equations
	\begin{align}
		u''_{0,k} &= \left[ \left( 2 + 2 \eps - 3 \delta + 2 \eps^2 - 3 \eps \delta + \delta^2 - \frac{\delta'}{\mathcal{H}} \right) \mathcal{H}^2 - k^2 \right] \\
		\matfett{\psi}''_k &= (2 + 2 \eps - 4 \delta) \mathcal{H} \matfett{\psi}'_k + \left[ \left( 2 \eps - \delta + 2 \eps^2 - \eps \delta - \delta^2 - \frac{\delta'}{\mathcal{H}} \right) \mathcal{H}^2 - k^2 \right] \matfett{\psi}_k
	\end{align}
	and the constraint
	\begin{align}
		\matfett{\psi}'_k = - (\eps \rho + \delta) \mathcal{H} \matfett{\psi}_k + \eps \mathcal{H} u_{\rho,k} .
	\end{align}
	It takes the form
	\begin{align}
		u''_{\rho,k} =&\! \left[ \left( 2 \eps^2 \rho^2 + 4 (1 + \eps - \delta) \eps \rho + 2 + 2 \eps - 3 \delta + 2 \eps^2 - 3 \eps \delta + \delta^2 - \frac{\delta'}{\mathcal{H}} \right) \mathcal{H}^2 - k^2 \right] u_{\rho,k} \nonumber\\
		&+ 2 \big[ -\eps (1 + \eps) \rho + (\eps - \delta)(1 + \eps - \delta) \big] \mathcal{H}^2 \rho \matfett{\psi}_k \label{urhodyneq}
	\end{align}
	
	As the secondary constraint can be cast in the form
	\begin{align}
		\matfett{\psi}_k = \frac{\eps \mathcal{H}}{k^2} \big( - u'_{\rho,k} + (\eps \rho + 1 + \eps - \delta) \mathcal{H} u_{\rho,k} \big)
	\end{align}
	one readily sees that the coefficient of $\matfett{\psi}_k$ in equation \eqref{urhodyneq} must vanish to satisfy requirement (b), as orders of $k$ translate to spatial derivatives. The condition $-\eps (1 + \eps) \rho^2 + (\eps - \delta)(1 + \eps - \delta) \rho =0$ implies the claim.
\end{proof}

Although $u = \matfett{\psi} + \lambda$ is not uniquely determined by theorem \ref{theorem:canvar}, it is clear that the alternative solution implies restrictions on $\eps$ and $\delta$. In this sense, $u$ is the unique suitable choice for all background configurations. We will point out one feature of the alternative solution here which is especially unwanted in the context of inflationary models.

\begin{Lemma} $ $\\
	If $\delta = \mathcal{O}(\eps)$, then $\rho = \frac{(\eps - \delta)(1 + \eps - \delta)}{\eps (1+ \eps)}$ implies $(\eps - \delta)' = (\eps - \delta) \mathcal{H} + \mathcal{O}(\eps^2)$.
\end{Lemma}
\begin{proof} $ $\\
	Using relation \eqref{drho} one can derive
	\begin{align}
		- \frac{\eps'}{\eps} - \frac{\eps'}{1+\eps} + \frac{(\eps-\delta)'}{\eps-\delta} + \frac{(\eps-\delta)'}{1+\eps-\delta} = \frac{\rho'}{\rho} = \frac{(\eps-\delta)(1+\eps-\delta)}{1+\eps}\mathcal{H} + (1+\eps) \mathcal{H}
	\end{align}
	which can be recast using $\gamma : = \eps - \delta$ to yield
	\begin{align}
		\gamma' = \frac{\gamma(1+\gamma)\big(\gamma(1+\gamma) + (1+\eps)^2 + 2\gamma(1+2\eps)\big)}{(1+\eps)(1+2\gamma)} \mathcal{H}
	\end{align}
	The claim follows by Taylor expansion of this expression.
\end{proof}

If one invokes the slow roll assumption and assumes the higher order terms to be small, one is left with a readily solvable differential equation, that yields $\eps - \gamma \propto a$. Thus, contrary to the assumption, $\eps$ and $\gamma$ are not small and slowly varying, which shows that the restrictions on $\eps$ and $\delta$ imposed by the alternative solution are incompatible with slow roll inflationary models. We will therefore discard this solution and regard $X=\matfett{\psi} + \lambda$ as a preferred solution. This result gives a justification, from a more structural point of view, for the usual choice of canonical variable in the standard literature on cosmology.

To generalize the result presented here to a slightly more general description of a scalar field, one might relax the first requirement by allowing $P_k = D X'_k$ where $D$ is some smooth function of time, but does not depend on $k$. One would expect that this introduces at least a freedom of the form $X=F \cdot (\matfett{\psi} + \lambda)$ for the preferred solution where $F$ is some smooth function of time.

The standard approach to quantization assuming CCR for the field $u = \lambda + \matfett{\psi} = a \chi + z \Psi$ was shown to be justified. It is thus interesting to investigate what the CCR quantization of $u$ implies for the commutation relations of $\Psi$ and $\chi$.
\begin{Lemma} $ $ \\
	If standard CCR for $u$ are imposed, $\Psi$ and $\chi$ are non-local.
\end{Lemma}
\begin{proof} $ $ \\
	As a first step we assume at least a basic consequence of locality to hold,
	\begin{align}
		\left[\Psi(x), \Psi(y) \right]\big|_{x^0=y^0} &= \left[\Psi'(x), \Psi'(y) \right]\big|_{x^0=y^0} = 0 .
	\end{align}
	Using equations \eqref{constr1} and \eqref{constr2} the following relations can be derived
	\begin{align}
		u =& \frac{z}{\eps \mathcal{H}} \Psi' + \frac{1+ \eps}{\eps} z \Psi \displaybreak[0]\nonumber\\
		u' =& \frac{z}{\eps \mathcal{H}} \triangle \Psi + (1+\eps-\delta) \mathcal{H} u \displaybreak[0]\nonumber\\
		[u(x), u'(y)]\big|_{x^0=y^0} =& - \frac{z^2}{\eps^2 \mathcal{H}^2} \triangle_y [\Psi(y), \Psi'(x)]\big|_{x^0=y^0} \label{CCRrel}
	\end{align}
	Now the last relation can be recast in the form
	\begin{align*}
		\frac{i z^2}{\eps^2 \mathcal{H}^2} \triangle_y [\Psi(y), \Psi'(x)]\big|_{x^0=y^0} =& \delta(\vec{x}-\vec{y}) = \delta(\vec{y}-\vec{x})
	\end{align*}
	which is the distributional differential equation for the fundamental solution of the Laplacian. This is straightforward to solve and yields
	\begin{align*}
		[\Psi(x), \Psi'(y)]\big|_{x^0=y^0} =& \frac{i \eps^2 \mathcal{H}^2}{4\pi z^2} \frac{1}{|\vec{x}-\vec{y}|}
	\end{align*}
	such that the result of lemma \ref{Lemma:CCRNoGo} can be readily applied
	\begin{align}
		\left[\chi(x), \chi(y) \right]\big|_{x^0=y^0} &= \left[\chi'(x), \chi'(y) \right]\big|_{x^0=y^0} = 0\\
		\left[\chi(x), \chi'(y) \right]\big|_{x^0=y^0} &= \frac{i}{a^2} \delta(\vec{x}-\vec{y}) - \frac{i \eps \mathcal{H}^2}{4\pi a^2} \frac{1}{|\vec{x}-\vec{y}|}
	\end{align}
This proves the claim.
\end{proof}

This result has intriguing implications, as $\Psi$ and $\chi$ have a clear physical interpretation but the ``natural'' quantization procedure suggests that they are non-local fields. Such a situation may point to an underlying non-commutativity of spacetime, as the cause of this non-locality. It is worth noting that ``non-commutative inflation'', i.e. inflationary cosmological scenarios based on a non-commutativity of spacetime, is a topic of evolving interest during the past years, see e.g. \cite{ABM}\cite{BK}\cite{DPR}\cite{MO}. The above result may be interpreted as indicating that it is ``natural'' to assume spacetime non-commutativity in inflationary scenarios.

\section{Weyl Quantization} \label{sec:weylq}

In this section the algebraic Weyl quantization procedure is applied to the dynamic system of perturbations. This application of a rigorous quantization scheme serves to clarify the choice of quantum state, which is obfuscated by the generic treatment using creation an annihilation operators. We are not aware that this approach has been applied to the system at hand, but the application is straightforward and not original. However, it serves to identify a possible target for further investigation.

As a general reference for algebraic methods in quantum field theory see \cite{Haa}, for the application of Weyl quantization in quantum field theory on curved spacetime see \cite{Wal94}. An application of the Weyl algebra in the context of loop quantum gravity is presented in \cite{Thi}. The Weyl quantization procedure has the benefit that the operators involved are bounded if represented on a Hilbert space, which makes it possible to define a C${}^*$ norm on the Weyl algebra.

For the quantization procedure which is applied here it is of great importance that the constraints are second class constraints. For systems with first class constraints quantization is much more complicated \cite{Giu} and methods such as refined algebraic quantization \cite{ALMMT} and more recently a functional approach to the Batalin-Vilkovisky formalism \cite{FR1}\cite{FR2} have been developed. For the case at hand a rather conventional approach suffices.

The general pattern that will be followed is the same as in the previous subsections. First the system without constraint\footnote{We speak only of one constraint here, as the second constraint is derived from this constraint as a secondary constraint.} is considered and then the impact of the constraint on this system is investigated. The present subsection does not provide additional information on the dynamic system but clarifies the impact of the constraint at the level of algebraic quantum field theory.

As we will deal with the quantum system here and not the classical system, the question whether quantization commutes with reduction comes into play. In the present setting, reduction on the classical level can be split up into two parts. First the Poisson bracket on phase space is replaced by the Dirac bracket, which induces a presymplectic form, and in the second step a projection to the symplectic subspace is performed. It is not clear how the first part of the reduction procedure should be implemented at the quantum level, while the second part of the reduction procedure gives rise to a well defined algebraic procedure, such that quantization commutes with the second part of reduction by construction.

Starting out from the dynamic system for $\Psi$ and $\chi$ is unsuitable for a consideration from the point of view of algebraic quantum field theory as the system contains coupling of the fields, which undermines a completely rigorous treatment with current methods. Instead it would be preferable to start from a system of two independent free scalar fields. To achieve this, it is necessary to apply the constraints to arrive at a dynamic system for $\matfett{\Psi} = \tfrac{a^2}{\mathcal{H}z} \Psi$ and $u$.
\begin{alignat}{3}
	\matfett{\Psi}'' - && \triangle \matfett{\Psi} + && \left( - 2 \eps + \delta + \eps \delta - \delta^2 - \frac{\delta'}{\mathcal{H}} \right) \mathcal{H}^2 \matfett{\Psi} &=0\\
	u'' - && \triangle u + && \left( - 2 - 2 \eps + 3 \delta - 2 \eps^2 + 3 \eps \delta - \delta^2 + \frac{\delta'}{\mathcal{H}} \right) \mathcal{H}^2 u &= 0
\end{alignat}
and the constraint takes the form
\begin{align}
	\matfett{\Psi}' + ( 1 + \eps - \delta ) \mathcal{H} \matfett{\Psi} = \frac{\eps a^2}{z^2} u = 4 \pi G u
\end{align}

Ignoring the constraint, $\matfett{\Psi}$ and $u$ are free scalar fields and the symplectic space $\mathscr{P}$ of their modes is straightforward. Now in the present treatment the reduction of phase space is done in two steps. First the Poisson bracket is replaced by the degenerate Dirac bracket, leading to a space $\mathscr{D}$ with a symplectic subspace. Then a projection to the maximal symplectic subspace is performed by factoring out the primary and secondary constraint, leading to a symplectic space $\mathscr{D}/\mathcal{C}$.

To tackle the problem of quantization, it is useful to take the spatial Fourier transform of the fields, as this allows to make direct contact with the previous treatment. In this context, the Dirac form is defined as in equation \eqref{IntDiracDef}, where the mode Dirac form is of course different from the one given in equation \eqref{DiracForm}, as we are working with different coordinates on the symplectic space here. However, it can be expected that the limit $k\rar0$ should be of the same order of divergence as in equation \eqref{DiracForm} because the coordinates are related by a $k$-independent linear transformation. Therefore the integral can be expected to exist. The constraint in mode form is
\begin{align*}
	\mathcal{C}_{1,k} := \widetilde{\matfett{\Psi}}_k' + ( 1 + \eps - \delta ) \mathcal{H} \widetilde{\matfett{\Psi}}_k - 4 \pi G \widetilde{u}_k = 0
\end{align*}
and the Dirac bracket of the constraint with an arbitrary initial value set is
\begin{align}
	\left\{ \left( \begin{array}{c} f \\  (1 + \eps - \delta ) \mathcal{H} f \\ 0 \\ - 4 \pi G f \end{array} \right) , \left( \begin{array}{c} g_1 \\ g_2 \\ g_3 \\ g_4 \end{array} \right) \right\}_{\text{D}} &= 0 \label{ConstraintDirac}
\end{align}

The algebras of observables for $\matfett{\Psi}$ and $u$, the Weyl algebras $\mathfrak{A}_{\matfett{\Psi}/u}$, can be constructed in the standard way. The algebra of the full system is then given as the tensor product algebra $\mathfrak{A} = \mathfrak{A}_{\matfett{\Psi}} \otimes \mathfrak{A}_{u}$. Now the constraint cannot be simply interpreted as a relation in this algebra that has to be factored out to restrict to the constrained system, because it is incompatible with the commutation relations. As discussed above, applying the constraint is not very complicated at the level of field modes; at the level of the Weyl algebra the Weyl relations have to be replaced with the modified relations derived from the mode Dirac forms, profoundly changing the algebra.

For the thus acquired algebra $\mathfrak{A}_{\mathscr{D}}$ corresponding to $\mathscr{D}$ the constraint defines a non-trivial ideal, as the Dirac form is constructed to be degenerate in exactly this sense. This is in contrast to $\mathfrak{A}$ which, as a Weyl-Algebra corresponding to a non-degenerate Poisson form, is simple. The Weyl operators are defined with respect to a set of mode initial values
\begin{align*}
	\widehat{W} (\widetilde{\matfett{\Psi}}',\widetilde{\matfett{\Psi}}, \widetilde{u}', \widetilde{u} )
\end{align*}
and the Dirac-Weyl relation is
\begin{align}
	&\widehat{W} (f_1, f_2, f_3, f_4) \widehat{W} (g_1, g_2, g_3, g_4) \nonumber\\
	&\quad = \exp\left( \frac{i}{2} \left\{ \left( \begin{array}{c} f_1 \\ f_2 \\ f_3 \\ f_4 \end{array} \right) , \left( \begin{array}{c} g_1 \\ g_2 \\ g_3 \\ g_4 \end{array} \right) \right\}_{\text{D}} \right) \widehat{W} (f_1 + g_1, f_2 + g_2, f_3 + g_3, f_4 + g_4) \label{DiracWeylRelation}
\end{align}
This makes it possible to formulate the constraint at the level of the Dirac-Weyl algebra $\mathfrak{A}_{\mathscr{D}}$ as
\begin{align*}
	\forall f \in \mathscr{S}(\mathbb{R}^+) : \quad \widehat{W} (f, (1 + \eps - \delta ) \mathcal{H} f, 0, - 4 \pi G f) = \id
\end{align*}
where $\mathscr{S}(\mathbb{R}^+)$ denotes the space of Schwartz functions. The constraint is compatible with the Dirac-Weyl relation \eqref{DiracWeylRelation} by virtue of equation \eqref{ConstraintDirac}. The secondary constraint can be calculated and cast into a similar form at the level of the algebra. The constraint relations define two ideals on the algebra and are factored out. The resultant algebra is denoted as $\mathfrak{A}_{\mathscr{D}}/\mathcal{C}'$ and corresponds by construction to the reduced phase space $\mathscr{D}/\mathcal{C}$. This is illustrated in the following diagram, where it should be noted that the connection between $\mathfrak{A}$ and $\mathfrak{A}_{\mathscr{D}}$ is not clear at the purely algebraic level. Instead $\mathfrak{A}_{\mathscr{D}}$ is constructed using results at the classical level.

\begin{center}
	\begin{tikzcd}
		\mathscr{P} \arrow{rr}{\text{Dirac}} \arrow{d}{\text{Quantisation}} && \mathscr{D} \arrow{rr}{\text{Reduction}} \arrow{d}{\text{Quantisation}} && \mathscr{D}/\mathcal{C} \arrow{d}{\text{Quantisation}} \\
		\mathfrak{A} \arrow[dashrightarrow]{rr}{\text{?}} && \mathfrak{A}_{\mathscr{D}} \arrow{rr}{\text{Reduction}} && \mathfrak{A}_{\mathscr{D}}/\mathcal{C}'
	\end{tikzcd}
\end{center}

The commutation of quantization and reduction has been investigated in a plethora of models with a variety of tools. For a precedent to the type of commutative diagram presented here see e.g. \cite{LW} where a gauge theoretic model with first class constraints is investigated applying much more sophisticated technical tools than are used in the present work.

States from the state space $\mathcal{S}$ of the tensor product algebra $\mathfrak{A}$ can be freely chosen, for instance it is possible to pick states which have the Hadamard property for both fields. The simplest states of this kind are simply tensor products of Hadamard states for the single field algebras $\mathfrak{A}_{\matfett{\Psi}/u}$. When the constraint $\mathcal{C}'$ is applied, the resultant algebra $\mathfrak{A}_{\mathscr{D}}/\mathcal{C}'$ has a very different state space $\mathcal{S}_{\text{c}}$. In the construction applied here, the isomorphism $\mathfrak{A}_{\mathscr{D}}/\mathcal{C}' \simeq \mathfrak{A}_{u}$ is implied, which permits a canonical injective, unit preserving $C^*$-homomorphism $\alpha : \mathfrak{A}_{\mathscr{D}}/\mathcal{C}' \rar \mathfrak{A}$. The dual map to this homomorphism is a positive map $\alpha' : \mathcal{S} \rar \mathcal{S}_{\text{c}}$ which means that states on $\mathfrak{A}$ give rise to states on $\mathfrak{A}_{\mathscr{D}}/\mathcal{C}'$ by simply restricting to their action to $\mathfrak{A}_{u} \subset \mathfrak{A}$.

The explicit form of the relations between states in $\mathcal{S}$ and states in $\mathcal{S}_{\text{c}}$ is greatly simplified by our choice of dynamical variables. If one chooses different dynamical variables, like $\Psi$ and $\chi$, the construction of a corresponding algebra $\mathfrak{A}'$ is much more troublesome and a simple relation between $\mathfrak{A}'$ and $\mathfrak{A}'_{\mathscr{D}}/\mathcal{C}' \simeq \mathfrak{A}_{\mathscr{D}}/\mathcal{C}'$ need not exist.

\section{Conclusion}

We have illustrated how the fact that the dynamic system of linear scalar cosmological perturbations has only one degree of freedom makes it impossible to find a quantization procedure such that the metric and field perturbation both satisfy canonical commutation relations. By means of the Dirac formalism for the treatment of dynamic systems with constraints the standard literature choice for the canonical variables was justified from purely formal arguments. The Weyl quantization of the system is tractable and leads to a generic scalar field algebra. The physical degrees of freedom, namely the metric and field perturbations emerge as non-local fields from the quantization procedure. Although the approach of the present article could be applied to higher order perturbations, such an investigation provides a formidable challenge.

\section*{Acknowledgement}

This work was conducted as part of a Ph.D. project under supervision of Prof. R. Verch \cite{DoE} and with financial support of the International Max Planck Research School at the Max Planck Institute for Mathematics in the Sciences. We thank Adam Reichold for helpful discussions. We also thank two anonymous reviewers, whose comments helped improve this article.


\begin{thebibliography}{99}

\bibitem{ABM}
S. Alexander, R. H. Brandenberger, J. Magueijo \textit{Non-Commutative Inflation}, Phys. Rev. D 67, Issue 8 (2003) arXiv:hep-th/0108190
\bibitem{ALMMT}
A. Ashtekar, J. Lewandowski, D. Marolf, J. Mour\~ao, T. Thiemann \textit{Quantization of diffeomorphism invariant theories of connections with local degrees of freedom}, J. Math. Phys. 36 p. 6456 (1995) arXiv:gr-qc/9504018
\bibitem{BMR}
K. Bhattacharya, S. Mohanty, R. Rangarajan \textit{Temperature of the inflaton and duration of inflation from WMAP data}, Phys. Rev. Lett. 96, 121302 (2006) arXiv:hep-ph/0508070
\bibitem{BFM}
R. H. Brandenberger, H. A. Feldman, V. F. Mukhanov \textit{Theory of Cosmological Perturbations}, Phys. Rep. 215, p. 203--333 (1992)
\bibitem{BK}
R. H. Brandenberger, S. Koh \textit{Cosmological Perturbations in Non-Commutative Inflation} JCAP 2007, Issue 06 (2007) arXiv:hep-th/0702217
\bibitem{BuD}
T. S. Bunch, P. Davies \textit{Quantum Field Theory In De Sitter Space: Renormalization By Point Splitting}, Proceedings of the Royal Society of London, Series A, Vol. 360, No. 1700, p. 117--134 (1978)
\bibitem{DLP}
A. Dapor, J. Lewandowski, J. Puchta \textit{QFT on quantum spacetime: a compatible classical framework}, arXiv:gr-qc/1302.3038
\bibitem{DPP}
C. Dappiaggi, N. Pinamonti, M. Porrmann \textit{Local causal structures, Hadamard states and the principle of local covariance in quantum field theory}, Commun. Math. Phys. 304, p. 459--498 (2011) arXiv:hep-th/1001.0858
\bibitem{DEa}
P.D. D'Eath \textit{On the Existence of Perturbed Robertson-Walker Universes}, Ann. Phys. 98, p. 237--263 (1976)
\bibitem{Der}
A. Deriglazov \textit{Classical Mechanics: Hamiltonian and Lagrangian Formalism}, Springer (2010)
\bibitem{Dir}
P. A. M. Dirac \textit{Lectures on Quantum Mechanics}, Belfer Graduate School of Science, New York (1964)
\bibitem{Dod}
S. Dodelson \textit{Modern Cosmology}, Academic Press (2003)
\bibitem{DPR}
R. Durrer, H. Perrier, M. Rinaldi \textit{Explosive particle production in non-commutative inflation}, JHEP 2013, Issue 01 (2013) arXiv:gr-qc/1210.5373
\bibitem{DoE}
B. Eltzner \textit{Local Thermal Equilibrium on Curved Spacetimes and Linear Cosmological Perturbation Theory}, Ph.D. thesis, Leipzig (2012)\\
http://www.qucosa.de/recherche/frontdoor/?tx\_slubopus4frontend[id]=11747
\bibitem{FR1}
K. Fredenhagen, K. Rejzner \textit{Batalin-Vilkovisky formalism in the functional approach to classical field theory}, Comm. Math. Phys. 314, Issue 1, p. 93--127 (2012) arXiv:math-ph/1101.5112
\bibitem{FR2}
K. Fredenhagen, K. Rejzner \textit{Batalin-Vilkovisky formalism in perturbative algebraic quantum field theory}, Comm. Math. Phys. 317, Issue 3, p. 697--725 (2013) arXiv:math-ph/1110.5232
\bibitem{Giu}
D. Giulini \textit{That strange procedure called quantisation}, Lect. Notes Phys. 631, p. 17--40 (2003) arXiv:quant-ph/0304202
\bibitem{Haa}
R. Haag \textit{Local Quantum Physics: Fields, Particles, Algebras}, 2nd Edition, Springer (1996)
\bibitem{HT}
M. Henneaux, C. Teitelboim \textit{Quantization of Gauge Systems}, Princeton University Press (1992)
\bibitem{Hol}
S. Hollands \textit{Correlators, Feynman diagrams, and quantum no-hair in deSitter spacetime}, arXiv:gr-qc/1010.5367
\bibitem{KKLSS}
N. Kaloper, M. Kleban, A. Lawrence, S. Shenker, L. Susskind \textit{Initial conditions for inflation}, JHEP 0211, 037 (2002) arXiv:hep-th/0209231
\bibitem{KW}
B. S. Kay, R. M. Wald \textit{Theorems on the Uniqueness and Thermal Properties of Stationary, Nonsingular, Quasifree States on Spacetimes with a Bifurcate Killing Horizon}, Physics Reports 207, Issue 2, p. 49--136 (1991)
\bibitem{Kun}
S. Kundu \textit{Inflation with general initial conditions for scalar perturbations}, JCAP 1202, 005 (2012) arXiv:astro-ph/1110.4688
\bibitem{LW}
N. P. Landsman, U. A. Wiedemann \textit{The Stueckelberg-Kibble Model as an Example of Quantized Symplectic Reduction}, J. Math. Phys. 37, p. 2731--2747 (1996) arXiv:hep-th/9508134
\bibitem{Lo05}
B. Losic \textit{Quantum backreactions in slow-roll and de Sitter spacetimes}, Ph.D. thesis, University Of British Columbia (2005)\\
http://laplace.physics.ubc.ca/ThesesOthers/Phd/losic.pdf
\bibitem{LU08}
B. Losic, W. G. Unruh \textit{Cosmological Perturbation Theory in Slow-Roll Spacetimes}, Phys. Rev. Lett. 101, Issue 11 (2008) arXiv:gr-qc/0804.4296
\bibitem{MO}
U. D. Machado, R. Opher \textit{Generalized Non-Commutative Inflation}, Class. Quantum Grav. 29, Number 6 (2012) arXiv:astro-ph/1102.4828
\bibitem{MM}
D. Marolf, I. A. Morrison \textit{The IR stability of de Sitter QFT: results at all orders}, Phys. Rev. D 84, Issue 4 (2011) arXiv:gr-qc/1010.5327
\bibitem{Mon75}
V. Moncrief, \textit{Space-time symmetries and linearization stability of the Einstein Equations. I.}
J. Math. Phys. 16, Issue 3, p. 493--498 (1975)
\bibitem{Mon76}
V. Moncrief. \textit{Space-time symmetries and linearization stability of the Einstein Equations. II.}
J. Math. Phys. 17, Issue 10, p. 1893--1902 (1976)
\bibitem{PSS}
A. Perez, H. Sahlmann, D. Sudarsky \textit{On the quantum origin of the seeds of cosmic structure}, Class. Quantum Grav. 23, p. 2317--2354 (2006)
\bibitem{Str}
N. Straumann \textit{From primordial quantum fluctuations to the anisotropies of the cosmic microwave background radiation}, Annalen der Physik 15, No. 10--11, p. 701--845 (2006) arXiv:hep-ph/0505249
\bibitem{Thi}
T. Thiemann \textit{Loop Quantum Gravity: An Inside View}, Lecture Notes in Physics 721, p. 185--263 (2007) arXiv:hep-th/0608210
\bibitem{Un98}
W. G. Unruh \textit{Cosmological long wavelength perturbations},\\
arXiv:astro-ph/9802323
\bibitem{Wal77}
R. Wald \textit{The Back Reaction Effect in Particle Creation in Curved Spacetime}, Commun. Math. Phys. 54, p. 1 (1977)
\bibitem{Wal84}
R. Wald \textit{General Relativity}, The University of Chicago Press (1984)
\bibitem{Wal94}
R. M. Wald \textit{Quantum field theory in curved spacetime and black hole thermodynamics}, Chicago Lectures in Physics, The University of Chicago Press (1994)
\bibitem{Wei}
S. Weinberg \textit{Lectures on Quantum Mechanics}, Cambridge University Press (2012)

\end{thebibliography}
\end{document}